\documentclass[11pt,a4paper]{article}
\usepackage{modules/sevd}

\title{Sublinear Explicit Incremental Planar Voronoi Diagrams}

\author{
	Elena Arseneva\thanks{Saint Petersburg State University} \qquad
	John Iacono\thanks{Universit\'{e} libre de Bruxelles; New York University} \qquad
	Grigorios Koumoutsos\thanks{Universit\'{e} libre de Bruxelles} \\
	Stefan Langerman\samethanks[3] \qquad
	Boris Zolotov\samethanks[1]
}

\date{}

\begin{document}

\maketitle

\begin{abstract}
A data structure is presented that explicitly maintains the graph of a Voronoi diagram of $N$ point sites in the plane or the dual graph of a convex hull of points in three dimensions while allowing insertions of new sites/points. 
Our structure supports insertions in $\Ot (N^{3/4})$ expected amortized time, where $\Ot$ suppresses polylogarithmic terms. 
This is the first result to achieve sublinear time insertions; previously it was shown by Allen et~al.~that $\Theta(\sqrt{N})$ amortized combinatorial changes per insertion could occur in the Voronoi diagram but a sublinear-time algorithm was only presented for the special case of points in convex position.
\end{abstract}
	
\section{Introduction}

Voronoi diagrams and convex hulls are two keystone geometric structures of central importance to computational geometry.
We focus the description here on planar Voronoi diagrams of points; the results can be extended to the dual graph of 3D-convex hulls using a standard lifting trick. Several algorithms, based on various different techniques, have been developed over the years that compute the Voronoi diagram of a set of $N$ points in optimal time $O(N \log N)$~\cite{DBLP:books/el/00/AurenhammerK00}. Surprisingly however, the problem of maintaining dynamically a Voronoi diagram subject to insertions/deletions of points is not well understood. 

\paragraph{Incremental Voronoi Diagrams.} In this paper we focus on the problem of maintaining the Voronoi diagram under insertion of new sites. In Allen et al.\ \citAllen\ it was observed that while there could be a linear number of changes to the embedded Voronoi diagram with each site insertion, this is not equivalent to the number of combinatorial changes (i.e., edge insertions and deletions) to the graph of the Voronoi diagram. What is more, it was proved that the maximum number of combinatorial changes per site insertion is $\Theta(\sqrt{N})$ amortized. This opened the possibility to maintain the Voronoi diagram graph under insertions with sublinear update time. Allen et al.~\citAllen\ achieved this for the restricted case where the sites are in convex position, where they designed a data structure with $O(\sqrt{N} \log^7 N)$ amortized insertion time. This result relies crucially on the fact that the Voronoi diagram of a set of points in convex position is a tree. Other more restricted special cases have been studied where the number of combinatorial changes is $\Theta(\log N)$~\cite{DBLP:journals/algorithmica/AronovBDGILS18,Pet10}.

\paragraph{Our Result.} In this work, we provide a data structure that explicitly maintains the graph of a Voronoi diagram of arbitrary point sites in $\br^2$ while allowing insertions of new sites in $\Ot (N^{3/4})$ amortized time, where $\Ot$ suppresses polylogarithmic terms. This is the first data structure supporting insertions in sublinear time for this problem.

We crucially note that we are interested in maintaining explicitly the Voronoi diagram. In particular, we store the diagram as a graph in adjacency list format on which primitive operations, including links and cuts, are performed. This is different than just maintaining a data structure that answers nearest neighbor queries. This case can be solved dynamically in polylogarithmic time by Dynamic Nearest Neighbor (DNN) structures~\cite{DBLP:journals/jacm/Chan10,KaplanMRSS17,Chan19}; this relies heavily on the fact that nearest neighbor is a decomposable search problem, whereas maintaining the Voronoi diagram is clearly not. In fact, here we use those DNN structures as subroutines for solving our problem.

We remark that maintaining the Voronoi diagram in sublinear time in the fully dynamic setting (i.e., with both insertions and deletions) is hopeless as the
$\Theta\lrp{\sqrt{N}}$ amortized bound of combinatorial changes for insertions becomes $\Theta\lrp{N}$.

\subsection{Brief description of our approach}
We now give a high-level overview of our approach and the organization of the rest of the paper. 

We store the Voronoi diagram as a combinatorial graph, which allows to quickly retrieve any geometric information if needed. 

Suppose we wish to insert a new site $s_N$ into the diagram, and let $f$ be the cell of the diagram that contains $s_N$. This cell can be found in polylogarithmic time using DNN structures. It is known~\citAllen\ that the affected cells that need to be updated, i.e. that undergo \emph{combinatorial changes}, form a connected region including cell $f$. To update the diagram, we discover all the affected cells by a variation of the breadth-first search starting from $f$. 

\paragraph{Small and big cells.} The main high-level idea is to divide the cells of the Voronoi diagram into \emph{small} and \emph{big}: small are the ones that have at most $N^{1/4}$ vertices and big are the ones that have more. For a Voronoi cell $f$, by the \paws of $f$ we denote the Voronoi vertices connected to the boundary of $f$ by one edge. What we do to process small and big cells is different, but is based on the following fact. Given an affected cell $f$ the neighboring cells of $f$ that change can be identified by finding all the \paws of $f$ whose Voronoi circles contain the newly inserted site $s_N$.  

A small cell is small enough to be processed by simply checking all $O(N^{1/4})$ Voronoi circles of its \paws in a brute force way. This takes $O(N^{1/4} \polylog(N)) = \Ot(N^{1/4})$ time per small cell. Since the amortized number of cells changing is $O(\sqrt{N})$, it takes $\Ot(\sqrt{N} \cdot  N^{1/4})= \Ot(N^{3/4})$ amortized time to perform all updates involving small cells.

For each big cell with $b_i$ neighbors, we store a circular linked list of $\Theta(b_i/N^{1/4})$ data structures, each associated with the consecutive range of $O(N^{1/4})$ of its paws.
Each structure stores the Voronoi circles for those \paws that are relevant. 
These Dynamic Circle-Reporting structures (DCRs) are known variants of the DNN structure that support insertion and deletion of circles in polylogarithmic time, and given a query point,
report all $k$ circles containing the point in time $\Ot(k)$. 
Overall, operations involving a big cell require polylogarithmic time in the number of affected neighbors. Since there are at most $O(N^{3/4})$ big cells, the total time to process them is $O(N^{3/4} \polylog(N)) = \Ot (N^{3/4})$. 

Overall, we need $\Ot(N^{3/4})$ amortized time for updating both small and big cells, thus the main result follows.

\paragraph{Note on previous work.} In a preliminary version of this paper~\cite{AIKLZ19} we presented the same result using a randomized data structure and providing bounds on the expected running time; the randomization came solely from the shallow cuttings used in the DNN structure~\cite{DBLP:journals/jacm/Chan10}. Since the initial development of this work, Chan presented DNN structures that use shallow cuttings deterministically~\cite{Chan19} (implicit in~\cite{ChanT16}); using this structure all DNN and DCR structures used in this paper can be implemented deterministically. As a result, our overall structure is deterministic. 

\paragraph{Organization.} The rest of the paper is organized as follows. In Section~\ref{sec:defflarb} we characterize affected cells that undergo combinatorial changes. In Section~\ref{sec:descrDS} we give a detailed description of our data structure. In Section~\ref{sec:recogGeneral} we present the procedure to find all the affected cells, and in Sections~\ref{sec:implement} and~\ref{sec:smallBig} we describe the procedure to implement the combinatorial changes and update the data structure accordingly. 

\section{Preliminaries and definitions}

We begin with standard definitions related to Voronoi diagrams and their basic properties. 

	Let $S \coloneqq \{s_1, s_2, \ldots, s_N\}$ be a set of $N$ distinct points in $\br^2$; these points are called \emph{sites}. Let $\mathrm{dist} (\cdot, \cdot)$ denote the Euclidean distance between two points in  $\br^2$. We assume that the sites in $S$ are in {\itshape general position}, that is, no four sites lie on a common circle.

\begin{definition} The \emph{Voronoi diagram} of $S$ is the subdivision of $\br^2$ into $N$ {cells}, called \emph{Voronoi cells}, one cell for each site in $S$, such that a point $q$ lies in the Voronoi cell
of a site $s_i$ if and only if
	$$\mathrm{dist} (q, s_i) < \mathrm{dist} (q, s_j)$$
	for each $s_j \in S$ with $j \ne i$. 
\end{definition}

Let  $f_i$ denote the Voronoi cell of a site $s_i$. Edges of the Voronoi diagram, called \emph{Voronoi edges,} are portions of bisectors between two sites which are the common boundary of the corresponding Voronoi cells. \emph{Voronoi vertices} are points where at least three Voronoi cells meet. The \emph{Voronoi circle} of a Voronoi vertex $v$ is the circle passing through the sites whose cells are incident to $v$, see Figure~\ref{fig:voronoiCircle}. Vertex $v$ is the center of its Voronoi circle.

\begin{figure}[h] \centering
\begin{subfigure}[b]{0.46\textwidth} \centering
	\includegraphics[width=5.3cm]{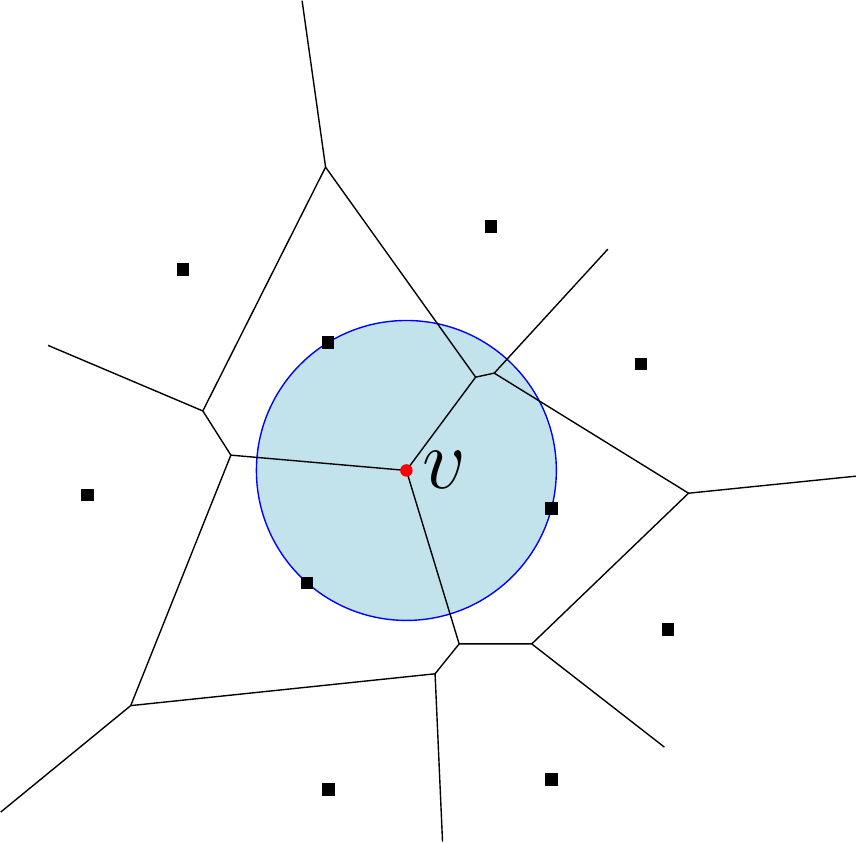}
	\caption{ } \label{fig:voronoiCircle}
\end{subfigure} ~~
\begin{subfigure}[b]{0.44\textwidth} \centering
	\includegraphics[width=5cm]{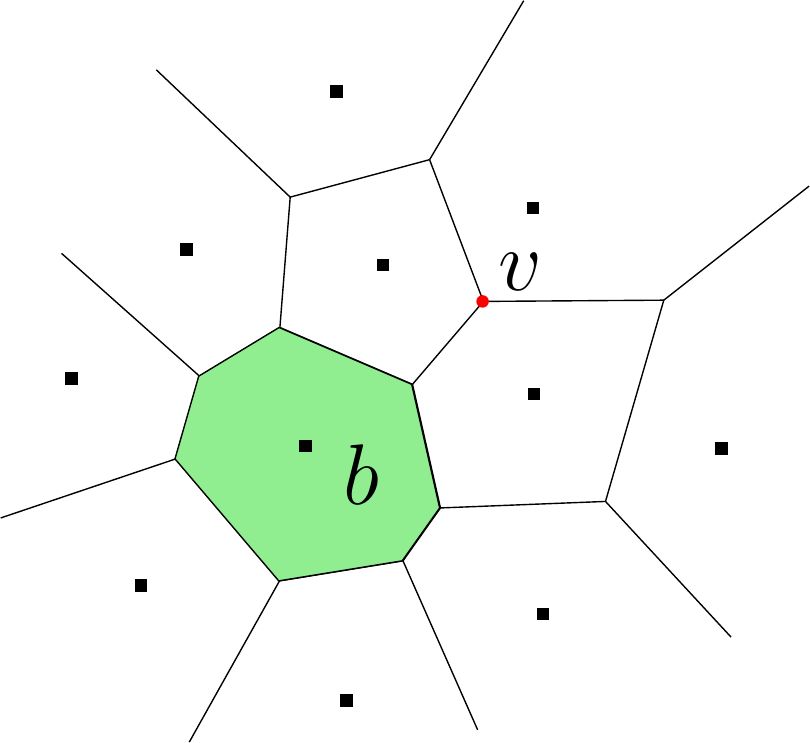}
	\caption{ } \label{fig:Paw}
\end{subfigure}
	\caption{A Voronoi diagram and (a) a Voronoi circle of vertex $v$ and (b) \paw $v$ of Voronoi cell $f$.}
\end{figure}

Since the sites are in general position, each Voronoi vertex has degree three. Each Voronoi edge is either a segment or a ray and the graph of the Voronoi diagram formed by its edges and vertices is planar and connected.

The next three definitions are specific to our data structure.

\begin{definition} \label{def:cellSize}
	The \emph{size} of a Voronoi cell is the number of Voronoi edges constituting its boundary. We denote the size of cell $f$ by $|f|$.
\end{definition}

\begin{definition} \label{def:bigCell}
	A Voronoi cell of a Voronoi diagram with $N$ sites is called a \emph{big cell} if it has size more than $\nof$. Otherwise it is called \emph{small}.
\end{definition}

\begin{definition} \label{def:Paw}
	 The \emph{\paws}of Voronoi cell $f$ are the Voronoi vertices that are connected to the boundary of $f$ by an edge and are not themselves on the boundary of $f$, see Figure~\ref{fig:Paw}. A \paw is called \emph{relevant} if it is not incident to a big cell.
\end{definition}

\subsection{Combinatorial changes to the Voronoi diagram and the flarb operation}
\label{sec:defflarb}

We now overview the definitions and results from Allen et al.~\citAllen\ that we need to present our approach. In order to prove the $\Theta(\noh)$ bound on the number of combinatorial changes caused by insertion of a site, a graph operation called \emph{flarb} is introduced.

Let $G$ be a planar 3-regular graph embedded in $\br^2$ without edge crossings (edges are not necessarily straight-line). Let $\mC$ be a simple closed Jordan curve in $\br^2$. 

\begin{definition} 
Curve $\mC$ is called \emph{flarbable} for $G$ if: \end{definition}

\begin{itemize}
	\item the graph induced by vertices inside the interior of $\mC$ is connected,
	\item $\mC$ intersects each edge of $G$ either at a single point or not at all,
	\item $\mC$ passes through no vertex of $G$, and
	\item the intersection of $\mC$ with each face of $G$ is path-connected.
\end{itemize}

For example, curve $\mC$ on Figure~\ref{fig:notflarb} is not flarbable since the intersection between its interior (shaded green) and the highlighted face (red) consists of two disconnected parts. In Figure~\ref{fig:flarb} curve $\mC'$ is flarbable.

\begin{figure}[h]
\centering
	\begin{subfigure}[t]{0.31\textwidth}
	\centering
	\includegraphics{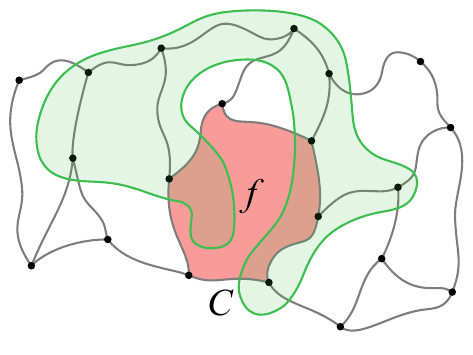}
	\caption{ }
	\label{fig:notflarb}
	\end{subfigure}
~
	\begin{subfigure}[t]{0.31\textwidth}
	\centering
	\includegraphics{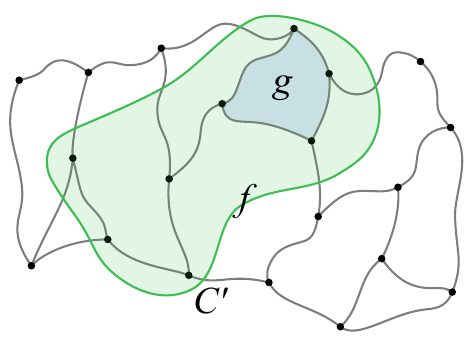}
	\caption{ }
	\label{fig:flarb}
	\end{subfigure}
~
	\begin{subfigure}[t]{0.31\textwidth}
	\centering
	\includegraphics{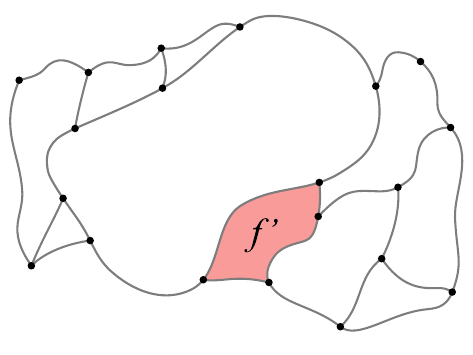}
	\caption{ }
	\label{fig:afterflarb}
	\end{subfigure}
\caption{Examples of (a) not flarbable curve (b) flarbable curve;
	     (c) result of applying the flarb operation
	     for curve $\mC'$.}
\label{fig:exflarb}
\end{figure}

Given a graph $G$ and a curve $\mC$ flarbable for $G$, the \emph{flarb} operation is, informally, removing part of $G$ that is inside $\mC$ and replacing it with $\mC$. Formally, the flarb operation for $G$ and $\mC$ is defined as follows (see Figure~\ref{fig:flarb},~\ref{fig:afterflarb}):

\begin{itemize}
	\item For each edge $e_i \in G$ that intersects $\mC$ let $u_i$ be its vertex lying inside $\mC$ and $v_i$ its vertex outside $\mC$. Create a new vertex $w_i = \mC \cap e_i$ and connect it to $v_i$ along $e_i$.
	\item Connect consecutive vertices $w_i$ along $\mC$.
	\item Delete all the vertices and edges inside $\mC$.
\end{itemize}

Let $\mG (G,\mC)$ denote the graph obtained by applying the flarb operation to graph $G$ and curve $\mC$.

\begin{prop} The following holds for graph $\mG ( G, \mC )$:
	(a) $\mG ( G, \mC )$ has at most two more vertices than $G$ does;
	(b) $\mG ( G, \mC )$ is a 3-regular planar graph;
	(c) $\mG ( G, \mC )$ has at most one more face than $G$ does.
\end{prop}

\begin{proof}
	Items (a) and (b) are proved in~\cite{DBLP:journals/dcg/AllenBIL17}, Lemma 2.2. To prove (c) note that there is one new face bounded by the cycle added along $\mC$ while performing the flarb. All the other faces of $G$ are either deleted, left intact, or cropped by $\mC$; these operations obviously do not increase the number of faces.
\end{proof}

\begin{theorem}[\citAllen] \label{thm:existsFlarb}
	Let $G$ be a graph of the Voronoi diagram of a set of $N-1$ sites $s_1 \ldots s_{N-1}$. For any new site $s_N$ there exists a flarbable curve $\mC$ such that the graph of the Voronoi diagram of sites $s_1 \ldots s_N$ is $\mG (G, \mC)$.
\end{theorem}

\paragraph{Cost of the flarb.} We want to analyze the number of structural changes that a graph undergoes when we apply the flarb operation to it. There are two basic combinatorial operations on graphs:

\begin{itemize}
	\item \emph{Link} is the addition of an edge between two non-adjacent vertices.
	\item \emph{Cut} is the removal of an existing edge.
\end{itemize}

Other combinatorial operations, for example insertion of vertex of degree 2, are assumed to have no cost.

\begin{definition}
	$\cost (G, \mC)$ is the minimum number of links and cuts needed to transform $G$ into $\mG (G, \mC)$.
\end{definition}

Note that sometimes there are less combinatorial changes needed than the number of edges intersected by $\mC$. Consider edges $e_1$, $e_2$ of $G$ crossed consecutively by $\mC$ and edge $n$ adjacent to them that reappears in $\mG (G, \mC)$ as a part $n^*$ of $\mC$. Then $n^*$ can be obtained without any links or cuts by lifting $n$ along $e_1$ and $e_2$ until it coincides with $n^*$ or (which is the same) shrinking $e_1$ and $e_2$ until their endpoints coincide with their intersections with $\mC$ (see Figure~\ref{fig:costfree}). We will call it {\itshape preserving operation}.

\begin{figure}[h] \centering
	\includegraphics[scale=1.1]{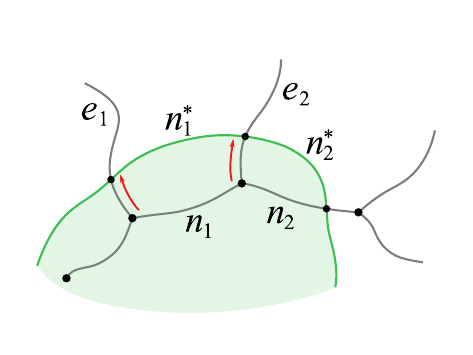}
	\caption{Edges $n_1^*$ and $n_2^*$ can be obtained without any links or cuts.}
	\label{fig:costfree}
\end{figure}

\begin{theorem}[\citAllen]
	For a flarbable curve $\mC$, it holds that \[ \cost (G, \mC) \le 12|\mathcal S (G, \mC)| + 3|\mathcal B (G, \mC)| + O(1). \]
	Where
	\begin{itemize}
		\item $|\mathcal B (G, \mC)|$ is the number of faces of $G$ wholly contained inside $\mC$ ($g$ is such a face on Figure~\ref{fig:flarb}).
		\item $|\mathcal S (G, \mC)|$ is the number of \emph{shrinking} faces—i.e. the faces whose number of edges decreases when flarb operation is applied (face $f$ is shrinking on Figures~\ref{fig:flarb}–\ref{fig:afterflarb}).
	\end{itemize}
\end{theorem}

The following upper bound can be used to evaluate the number of combinatorial changes needed to update the graph of a Voronoi diagram when a new site is inserted.

\begin{theorem}[\citAllen] \label{thm:amorCost} 
	Consider one insertion of a new site to a Voronoi diagram $V$\!\!.
\begin{itemize}
	\item The number of cells of $V$\!\! undergoing combinatorial changes is $O(\noh )$ amortized in a sequence of insertions;
	\item There are a constant number of combinatorial changes per cell;
	\item The cells of $V$\!\! with combinatorial changes form a connected region.
\end{itemize}
\end{theorem}

Further in this paper by {\itshape a change in cell} we always mean a combinatorial change, that is a {\itshape link} or a {\itshape cut}.

\section{Description of Data Structure} \label{sec:descrDS}

Our data structure consists of the following parts. \newcommand{\treename}{yard tree}

\begin{itemize}
	\item The graph $G_N$ of the Voronoi diagram represented by its adjacency list: for each Voronoi vertex $v$ we store the list of all Voronoi vertices connected to $v$. Since the sites are in general position, each list has length 3, therefore we can find and replace its elements in constant time. Thus any link or cut can be performed in constant time as well as insertion or deletion of a vertex of degree 2.
	\item For each vertex $v$ its {\itshape data} $D_v$ is stored. It is a list of the three sites that define the Voronoi circle of $v$, that is, the sites whose cells are incident to $v$.
	\item A dynamic nearest neighbor structure (DNN)~\cite{Chan19} for the sites which supports insertion and deletion of sites and nearest neighbor queries in $\Ot(1)$ amortized time. 
	\item The {\itshape graph $\Gamma_N$ of big cells} which is simply the dual graph to the subgraph of $G_N$ formed by big cells. Vertices of $\Gamma_N$ are big cells themselves and edges connect vertices corresponding to pairs of big cells that are adjacent. Graph $\Gamma_N$ has $O(\ntf)$ edges, since it is a planar graph of at most $\ntf$ vertices. For each pair of adjacent big cells $b_1$, $b_2$ we also store two Voronoi vertices they share. We store graph $\Gamma_N$ as an adjacency list, where for each vertex, its edges are stored in a binary search tree ordered counterclockwise around the corresponding big cell. The vertices of $\Gamma_N$ are stored in a binary search tree. This allows us to access any edge of $\Gamma_N$ in $\Ot(1)$ time.
	\item For each big cell $b_i$ store a circular linked list of $\Theta ( |b_i|\,/\,\nof )$ data structures each associated with a consecutive range of $O(\nof)$ \paws of $B_i$, see Figure~\ref{fig:consecDataStr}. Each structure stores the Voronoi circles of the relevant paws of $b_i$ (recall that a paw is relevant if it is not incident to a big cell, see Definition~\ref{def:Paw}).

\begin{figure}[h] \centering
	\includegraphics[width=6cm]{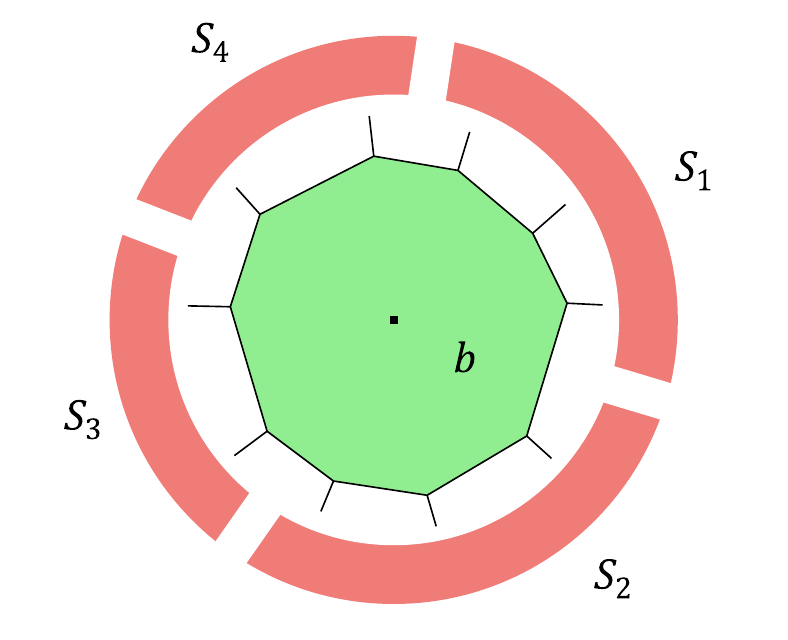}
	\caption{$b$ is a big cell; each of data structures $S_1 \ldots S_4$ is associated
	with a consecutive range of its \paws and stores Voronoi circles of the relevant ones.}
	\label{fig:consecDataStr}
\end{figure}

	The collections of circles are stored using \emph{dynamic circle-reporting structures (DCRs)} that are variants of the DNN structure constructed in~\cite{DBLP:journals/dcg/AllenBIL17}. DCRs support insertion and deletion of circles in time $\Ot(1)$, and given a query point, report all $k$ circles containing the point in time $\Ot(k)$.

	\item For each big cell $b_i$ a {\it\treename} $T_{b_i}$ supporting the following operations in $\Ot (1)$ time:
\begin{itemize}
	\item for a specified continuous range $v_1 \ldots v_m$ of vertices of $b_i$ updating $D_{v_1} \ldots D_{v_m}$, changing $s_i$ to a given site $s_j$.
	\item removing a continuous range of vertices from $b_i$ and create a new cell with these vertices preserving their order {\itshape (split)},
	\item merging the trees that correspond to big cells $b_i$ and $b_j$ (when two cells are merged their common edge is deleted), so that the same operations can apply to the resulting tree.
\end{itemize}
	One can use link-cut trees~\cite{ST83} or a collection of red-black trees with two-way pointers for this purpose, see Cormen et al.~\cite{cormen-iii} for details.

	\item For each cell $f_i$ we need to store its size $|f_i|$. 
\end{itemize}

\section{Insertion of a site}

We aim to implement the update of graph $G_{N-1}$ to become $G_N$ when a {\it new site} $s_N$ is added to the Voronoi diagram. Our goal is to quickly locate the affected cells that need combinatorial changes, and to avoid processing the other cells. When the cells that need changes are located, we implement these changes using the techniques of~\cite{DBLP:journals/dcg/AllenBIL17}.

Let the cell of the new site $s_N$ be called $f_N$. We denote the boundary of $f_N$ by $\mC_N$. According to Theorem~\ref{thm:existsFlarb}, what we are about to perform is the {\itshape flarb} operation on graph $G_{N-1}$ of the current Voronoi diagram and curve $\mC_N$.

\def\fdn{f_{\mathrm{dnn}}}

We first use the DNN structure to locate one Voronoi cell, call it $\fdn$, that must change—the one whose site is the closest to newly added $s_N$. We then add $s_N$ to the DNN. Finally we create the queue with all big cells of $G_{N-1}$ and cell $\fdn$. This whole procedure takes $\Ot (1)$ time as the list of all big cells is already stored.

We then remove each cell $f$ from the queue, process it, and add into the queue the small cells neighboring $f$ with unprocessed changes. 
We do not have to add  big cells neighboring $f$ as all of them were already in the initial queue and thus will be processed. 
Figure~\ref{fig:algInsertSite} shows a pseudocode of this procedure.

\def\Qar{\text{\tt Q}}
\def\Ch{\text{\tt Changed}}

\begin{figure}[h] \centering
     \hrule\medskip
\begin{algorithmic}[1]
	\State $\Qar \coloneqq \text{queue of all big cells}$
	\State $\Ch \coloneqq \text{empty array of cells}$
	\State $\Qar$.append\,$(\fdn)$
	\State $\text{DNN.insert}(s_N)$ \medskip
    \While{not $\Qar$.empty}
	\State $f \coloneqq \Qar.\text{dequeue}$
	\State $\Ch.\text{append}(f)$
	\If{$f$ is big}
	    \State add $f$'s small neighbors that need changes \par
		\hspace{0.95cm} to $\Qar$ using Section~\ref{sec:recogBig}
	\Else \hspace{6.25cm} $\Asterisk$ $f$ is small
	    \State add $f$'s neighbors that need changes \par
		\hspace{0.95cm} to $\Qar$ using Section~\ref{sec:recogSmall}
	\EndIf
    \EndWhile \medskip
	\State implement changes in cells in $\Ch$ \par
	    as described in Section~\ref{sec:implement}
	\State Some small cells have become big and some big have become small,\par
	    fix corresponding data structures as described in Section~\ref{sec:smallBig}
\end{algorithmic} \smallskip\hrule\medskip
	\caption{Pseudocode describing insertion of new site $s_N$}
	\label{fig:algInsertSite}
\end{figure}

\section{Recognizing cells with changes} \label{sec:recogGeneral}

Let $f$ be a cell with combinatorial changes. We can identify the neighboring cells of $f$ that change using the following theorem:

\begin{theorem}[\cite{DBLP:journals/dcg/AllenBIL17}] \label{thm:identChanges}
	Let $g$ be a cell adjacent to $f$. Let $v_1$, $v_2$ be the vertices of $g$ that are \paws of $f$. Cell $g$ needs to undergo combinatorial changes if and only if the Voronoi circle of $v_1$ or $v_2$ encloses $s_N$.
\end{theorem}

See Figure~\ref{fig:identChanges} for an example. Cell $f$ is a cell with changes, $n_1$ and $n_2$ are its paws. The Voronoi circle of $n_1$ encloses the new site $s_N$ and the one of $n_2$ does not. Therefore cells $f_1$ and $f_2$ need combinatorial changes as they are incident to vertex $n_1$, and cell $f_3$ does not need any changes.

We now consider separately the case when cell $f$ is a big cell (Section~\ref{sec:recogBig}) and the case when it is a small cell (Section~\ref{sec:recogSmall}).

\subsection{Cell $f$ is big} \label{sec:recogBig}

We use DCRs of cell $f$: they return all the relevant \paws of $f$ whose Voronoi circles enclose $s_N$. Small cells that are incident to these \paws and are adjacent to $f$ need combinatorial changes and thus have to be added to queue $\Qar$.

\begin{figure}[h] \centering
\begin{subfigure}[m]{0.48\textwidth} \centering
	\includegraphics[width=6.4cm]{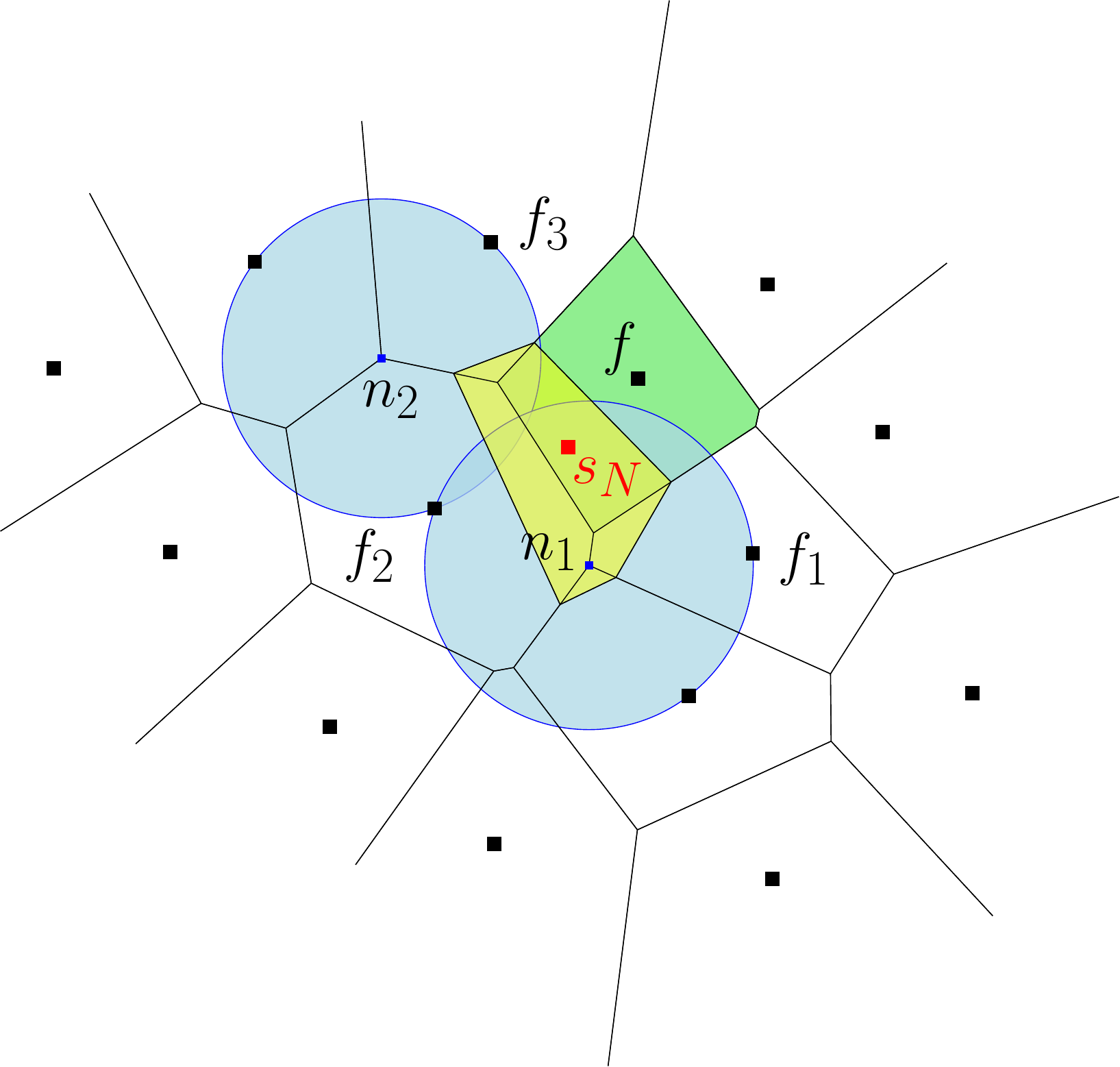}
	\caption{ }
	\label{fig:identChanges}
\end{subfigure} ~~
\begin{subfigure}[m]{0.42\textwidth} \centering
	\includegraphics[scale=1.1]{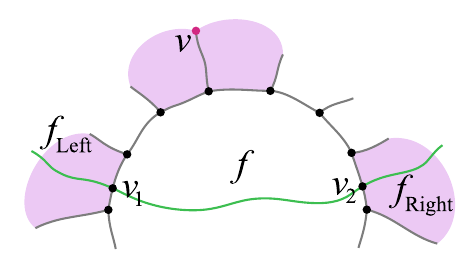}
	\caption{ }
	\label{fig:whomToAdd}
\end{subfigure}
	\caption{Identifying Voronoi cells that need changes. (a) Voronoi circle
	     of vertex~$n_1$ encloses~$s_N$, and the circle of~$n_2$ does not.
	     (b)~$v$ is a \paw of big cell~$f$ returned by a DCR. Highlighted are
	     the cells that are to be added to the queue.}
\end{figure}

Two cells are to be considered separately: those that are neighboring $f$ through an edge that is crossed by $\mC_N$. Denote them by $f_{\mathrm{Left}}$ and $f_{\mathrm{Right}}$, see Figure~\ref{fig:whomToAdd}. If they are small, we check whether the Voronoi circles of at most four \paws of $f$ incident to them (call these \paws $p_1 \ldots p_4$) enclose $s_N$, and, if yes, add the corresponding cell to the queue. To find Voronoi circles of these \paws we get the data $D_{p_1} \ldots D_{p_4}$ from the structures $T_{b_i}$ associated with big cells adjacent to $f_{\mathrm{Left}}$ and $f_{\mathrm{Right}}$, which requires \otl time.

\subsection{Cell $f$ is small} \label{sec:recogSmall}

We can look at every \paw $n_i$ of $f$ and identify those, whose Voronoi circle encloses $s_N$. This requires $\Ot(\nof)$ time in total. We add to $\Qar$ small cells adjacent to $f$ that are incident to these paws as they need changes according to Theorem~\ref{thm:identChanges}.

\section{Implementing combinatorial changes} \label{sec:implement}

In this section we describe how to implement combinatorial changes in a cell $f$ which lies in $\Ch$. We again consider separately the case when $f$ is big (Section~\ref{sec:procBig}) and the case when $f$ is small (Section~\ref{sec:procSmall}).

\subsection{Processing big cells} \label{sec:procBig}

Processing a big cell $f$ consists of the following four steps:

\paragraph{Updating the vertices.} We update a continuous range of $f$'s vertices $v_1 \ldots v_k$—we need to change their data $D_{v_1} \ldots D_{v_k}$ to indicate that these vertices are now incident to the cell of $s_N$ and not the cell of $s$. This can be done in $\Ot(1)$ time using the \treename~$T_f$.

\paragraph{Updating the graph of the Voronoi diagram.} First thing to do is a {\itshape link} along $\mC_N$ creating two vertices: vertex $v_1$ incident to $f$, $f_N$, $f_{\mathrm{Left}}$, and vertex $v_2$ incident to $f$, $f_N$, $f_{\mathrm{Right}}$ (see Figure~\ref{fig:whomToAdd}). After this {\itshape link} the part of $f$ inside $\mC_N$ becomes a part of the new cell $f_N$—luckily, all vertices of this part are already updated during the previous step.

There may be some big cells adjacent to $f$ that are already processed, creating other parts of the new cell. We have to join these parts together by cutting edges of $f$ that have portions inside $\mC_N$ and are incident to already processed big cells. Finding these edges in a straightforward way could be slow as $f$ can have a really large number of edges inside $\mC_N$ and we do not have enough time to look at each of them individually. Luckily, graph $\Gamma_{N-1}$ contains the information about edges shared by big cells. Thus we can in $\Ot(1)$ time find and delete edges incident to both $f$ and already processed big cells inside $\mC_N$. The edges shared by $f$ and other big cells inside $\mC_N$ will be deleted when these other big cells will be processed.

\def\fn{f_N^{(f)}}

\paragraph{Updating the graph of big cells.} Two operations we just carried out—split of $f$ by the new edge $v_1 v_2$ and joining of some cells that are parts of $f_N$—can change the set of big cells and add or cut some connections between them. However, $\Gamma_{N-1}$ can be updated accordingly in \otl time when such an operation is executed. It can be done as follows:

While undergoing a split, vertex $f$ falls apart into two vertices: $f'$ and $\fn$ (the latter represents a part of the new cell). They share newly created edge $v_1v_2$ of the Voronoi diagram. Note that the cells adjacent to $\fn$ form a continuous range of cells that were adjacent to $f$.

Thus we need \otl time to cut a continuous range from the binary search tree of cells adjacent to $f$, \otl time to add a new edge between $f'$ and $\fn$ to their binary search trees and \otl time to re-balance the binary search tree of all big cells.

Joining can be also done in \otl time. When two cells $f_1$, $f_2$ are joined we remove a node corresponding to $f_1$ from binary search tree of $f_2$ and vice versa, this takes \otl time. We then join the trees of $f_1$ and $f_2$ also in \otl time since cells that were adjacent to $f_1$ form a continuous range of cells adjacent to the new one.

\paragraph{Fixing data structures.} There are two data structures associated with $f$ that have to be considered:

\begin{itemize}
	\item $T_f$ can be updated in \otl time the same way we did with the graph of big cells.

	\item DCRs of relevant \paws\!\!: when big cells are joined or split, most of DCR-s stay intact. The only DCRs that need to be rebuilt are those whose range contains the endpoints of the edge that is either cut or added. Rebuilding of a DCR takes $\Ot(\nof)$ time since at most $\O(\nof)$ circles are stored there.
\end{itemize}

\subsection{Processing small cells} \label{sec:procSmall}

A small cell is different from a big cell in that we can consider every edge of it, and it will take us $\Ot(\nof)$ time. We will implement the combinatorial changes in $f$, and after this we update the DCRs of neighboring big cells.

We can in time $\Ot(\nof)$ distinguish whether $f$ has one, two, or more vertices inside the new cell (if they exist). Below we describe these three cases separately.

	\paragraph{One vertex inside the new cell.} (See Figure~\ref{fig:1inside-a}.) Let $f_i$, $f_k$ be the cells adjacent to $f$ that share an edge with $f$ inside $\mC_N$. Let those edges be called $e_i$, $e_k$ respectively.

	It is certain that neither $f_i$ nor $f_k$ have been processed yet: if $f_i$ is processed then there would be a vertex $v = \mC_N \cap e_i$. Then we have to create the face in the graph that is separated from $f$, $f_i$, $f_k$, bounded by $\mC_N$, and is a part of the new cell $f_N$.

\begin{figure}[h] \centering
	\includegraphics[scale=1.1]{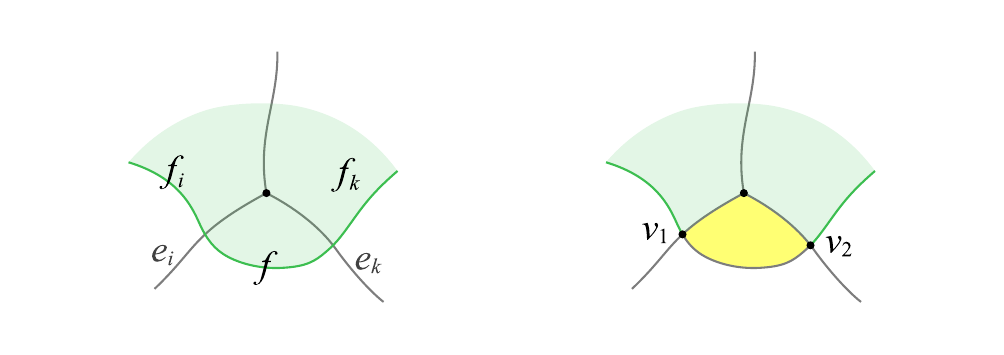}
	\abcaption{fig:1inside-a}{fig:1inside-b}
	\caption{Processing a cell with one vertex inside the new cell}
	\label{fig:1inside}
\end{figure}

	To do so, we perform a {\itshape link} operation inside $f$ along $\mC_N$: we create new vertices $v_1$ on $e_i$, $v_2$ on $e_k$ and add an edge $v_1v_2$ to $G_{N-1}$, see Figure~\ref{fig:1inside-b}. $v_1$ is incident to the cells of sites $s$, $s_i$ and $s_N$; $v_2$ is incident to the cells of $s$, $s_j$ and $s_N$.
	
	\paragraph{Two vertices inside the new cell.} We check whether cells adjacent to $f$ that share an edge with it inside $\mC_N$ has been already processed. If not (see Figure~\ref{fig:2inside-na}), we perform a {\itshape link} operation inside $f$ similarly to the previous paragraph, see Figure~\ref{fig:2inside-nb}.

\begin{figure}[h] \centering
	\includegraphics[scale=1.1]{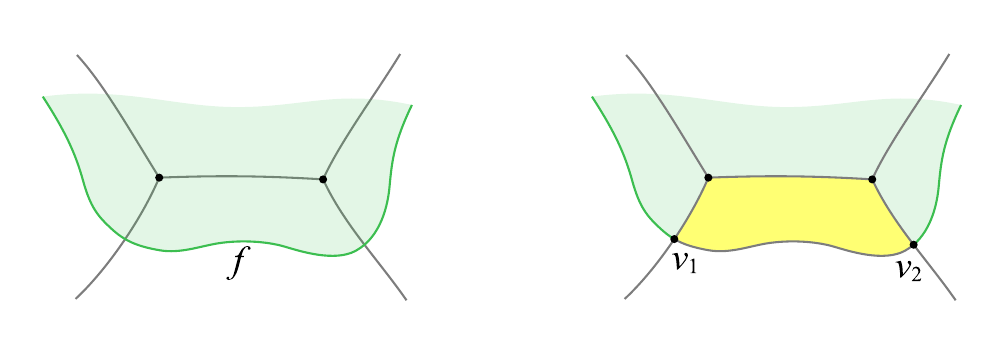}
	\abcaption{fig:2inside-na}{fig:2inside-nb}
	\caption{Processing a cell with two vertices inside the new cell when no surrounding cells are yet processed}
	\label{fig:2inside-n}
\end{figure}

Otherwise let us denote three faces sharing an edge with cell $f$ inside curve $\mC_N$ by $f_1$, $f_2$, $f_3$, see Figure~\ref{fig:2inside-ya}.

\begin{lemma}
	It is only $f_2$ that can have been already processed.
\end{lemma}

\begin{proof}
	Suppose $f_1$ is processed. It must then have an edge along $\mC_N$. It implies that there is a vertex where $\mC_N$ meets the common edge of $f_1$ and $f$. This vertex becomes the third vertex of $f$ inside $\mC_N$. However, $f$ has only two such vertices, which is a contradiction.
\end{proof}

\begin{figure}[h] \centering
	\includegraphics[scale=1.1]{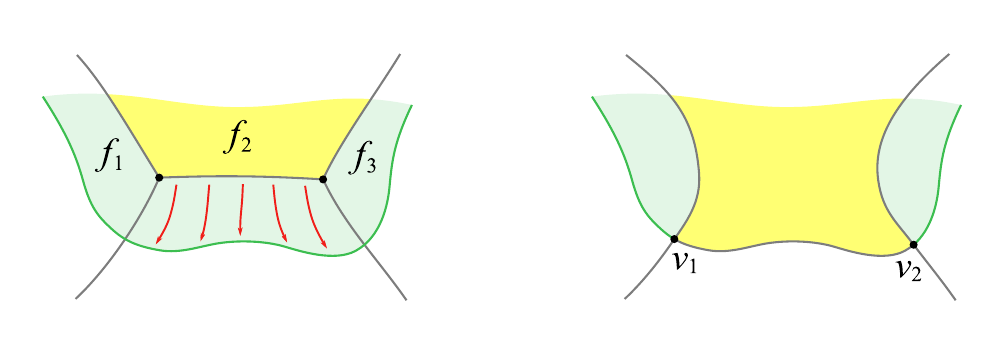}
	\abcaption{fig:2inside-ya}{fig:2inside-yb}
	\caption{Processing a cell with two vertices inside the new cell when there is a processed neighboring cells—no structural changes are needed}
	\label{fig:2inside-y}
\end{figure}

If $f_2$ is processed and is part of $f_N$ then the data $D_{v_1}$, $D_{v_2}$ of its vertices was updated when we were processing it. Therefore $f$ does not need to undergo any combinatorial changes, the common edge of $f$ and $f_N$ can be obtained by preserving operation which was described in Section~\ref{sec:defflarb}, see Figure~\ref{fig:costfree}. Thus no links and cuts are required, see Figure~\ref{fig:2inside-yb}. This completes implementing changes in $f$.

	\paragraph{More vertices inside the new cell.} Again we check whether any of the cells adjacent to $f$ has been processed already. If not, it is enough to perform one {\itshape link} creating two new vertices $v_1$, $v_2$ and to update the data $D_{v_j}$ of all the vertices of cell $f$ between $v_1$ and $v_2$: now they are incident to the cell of $s_N$, see Figure~\ref{fig:3inside-n}.

\begin{figure}[h] \centering
	\begin{subfigure}[t]{0.42\textwidth}
	\centering
	\includegraphics[scale=1.1]{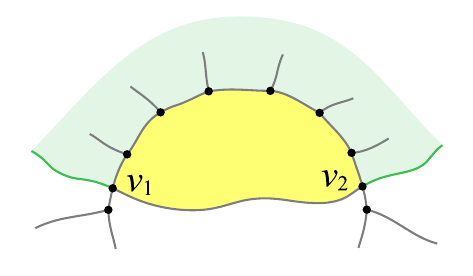}
	\caption{\ } \label{fig:3inside-n}
	\end{subfigure}
\quad
	\begin{subfigure}[t]{0.42\textwidth}
	\centering
	\includegraphics[scale=1.1]{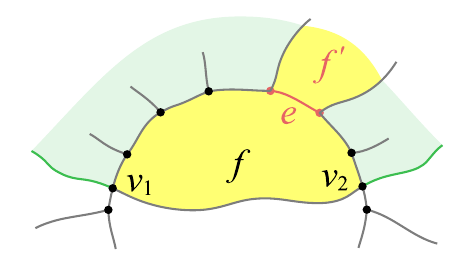}
	\caption{\ } \label{fig:3inside-y}
	\end{subfigure}
\caption{Processing a cell with three or more vertices inside the new cell. \\
	(a) No adjacent faces have been processed yet \\
	(b) Adjacent face $f'$ has been processed}
\label{fig:3inside}
\end{figure}

	If some cells sharing an edge with $f$ inside $\mC_N$ are already processed and represent a part of new cell, then for each processed cell $f'$ adjacent to $f$ we also perform a {\itshape cut} removing their common edge $e$ and then remove vertices incident to this edge that now have degree 2, see Figure~\ref{fig:3inside-y}.

\paragraph{Updating the DCRs of big neighbors of $f$.} The last step is that for every vertex $v$ of $f$ whose list of adjacent cells has changed during update of $G_{N-1}$ we find all big cells for which $v$ is a paw (there are at most three such cells, since $v$ has degree 3), recalculate the Voronoi circle of $v$, and update Voronoi circle of $v$ in DCRs of those cells which takes \otl time.

\section{When small cells become big} \label{sec:smallBig}

When a size of a cell crosses the threshold of $\nof$, it can be easily identified since we store all the sizes. If a big cell $b$ is split into a number of cells and one of them is small, or if $N$ becomes greater than $|b|^4$, we delete all the structures associated with it, including DCRs and the structure $T_f$. We also remove from $\Gamma_{N-1}$ the vertex corresponding to $b$.

Other way around, a new big cell can appear in the diagram when:
\begin{itemize}
	\item the new cell $f_N$ intersects many of old cells and has more than $\nof$ vertices, or
	\item a cell $f_k$ with $\nof-1$ vertices has one vertex inside $f_N$ and gets one additional vertex while being processed, see Figure~\ref{fig:1inside}.
\end{itemize}

New cell $f_N$ inherits the portion of its DCRs from its parts that previously were parts of big cells. The number of circles of \paws of previously small cells that need to be added to DCRs can be bounded from above by the size of a small cell times number of cells that undergo changes — that is,
	$$\noh\cdot\nof = \ntf.$$
	
The structure $T_{f_N}$ is inherited in part by $f_N$ from big cells that intersect curve $\mC_N$. The number of vertices that have to be added to $T_{f_N}$ after that is bounded from above by the number of combinatorial changes in current insertion.

The cell $f_k$ still has size $|f_k| \le 2\nof$, so \treename~$T_{f_k}$ can be built in $\Ot(\nof)$: it only takes time polylogarithmic in the size of the cell to add each vertex.

\section{Correctness and time complexity}

\begin{theorem} \label{thm:timeComplexity}
	Inserting a new site $s_N$ to our data structure and updating it requires $\Ot (\ntf)$ amortized time.
\end{theorem}

\begin{proof}
	Let $s$ be the number of small cells that change, and $b_1, b_2, \ldots, b_{|B|}$ be the big cells. Let $\ell_i$ be the number of circles returned by the DCR structures of $b_i$.
	
	All the operations on a small cell take $\Ot (\nof)$ time. For all the big cells together the operations on updating the graph structure and the graph of big cells require $\Ot(\ntf)$ total time. The number of DCR-s that have to be rebuilt is bounded from above by the number of changes in the graph.
	
	Finally, the amortized time complexity is
\[
\Ot \lrp{
s\nof +
\sum_{i=1}^{|B|}  \lrp{ \lrc{\frac{|b_i|}{\nof }}+\ell_i}
+ \ntf + s\nof
}.
\]

Since $s$ is $O(\noh)$ amortized \cite{DBLP:journals/dcg/AllenBIL17},  
$\sum_{i=1}^{|B|} |b_i| \leq N$, 
$|B|\leq \ntf$, and 
$\sum_{i=1}^{|B|} \ell_i \leq s \nof$, 
this is simply $\Ot(\ntf)$ amortized.

\end{proof}

\section{Discussion}

The problem of maintaining the convex hull of a set of points in $\br^3$ subject to point insertion can also be solved using our data structure. Namely, consider the dual problem of maintaining the intersection of a set of halfspaces. The two blocks of our data structure that are specific for Voronoi diagrams, translate in this setting as follows. To find the first face affected by the insertion (or report that it does not exist) we need to find the vertex extreme in the direction normal to the plane being inserted; if it is affected, then all the three incident faces are affected. We check whether a vertex is affected by determinimg the above/below relation between this vertex and the plane baing inserted. Thus Chan’s structure~\cite{Chan19} is again enough for our needs.

There remains a gap between the $\Ot (N^{3/4})$ expected amortized runtime of our structure and the $\Theta\lrp{\sqrt{N}}$ amortized number of combinatorial changes to the Voronoi diagram. Also, it would be interesting to get output-sensitive bounds, where the update time depends on the number of combinatorial changes. This was achieved by Allen et. al.~\cite{DBLP:journals/dcg/AllenBIL17} for points in convex position, where their update time is $O(K \log^7 N )$, where $K$ the number of combinatorial changes. We are unable to show this using our technique, due to the fact that we need to process all $\Theta (N^{3/4})$ big cells, no matter how many of them undergo combinatorial changes.
	
\section{Acknowledgments} This work was completed during the Second Trans-Siberian Workshop on Geometric Data Structures. We thank the organizers and staff of Russian Railways (\raisebox{-0.45ex}{\includegraphics{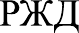}}) for creating an ideal research environment.

S.~L. is the Directeur de recherches du F.R.S-FNRS. J.~I. and G.~K. are supported by the Fonds de la Recherche Scientifique-FNRS under Grant no. MISU F 6001 1. E.~A. and B.~Z. are partially supported by the Foundation for the Advancement of Theoretical Physics and Mathematics ``BASIS'' and by ``Native towns'', a social investment program of PJSC ``Gazprom Neft''. J.~I. is supported by NSF grant CCF-1533564.

\vspace{-8pt}

\bibliographystyle{abbrv}
\bibliography{chan,john,boris}
	
\end{document}